%% file: duopoly_G.tex
\documentclass[11pt, a4paper]{article}
\usepackage{CJKutf8}
\usepackage{geometry}
\usepackage{subfigure}

\input{macro.tex}

\usepackage{authblk}
\usepackage{listings}

\begin{document}
\begin{CJK}{UTF8}{gbsn}

\title{Cournot duopoly games with isoelastic demands and diseconomies of scale}

\author{Xiaoliang Li\thanks{Corresponding author: xiaoliangbuaa@gmail.com}}
\affil{School of Finance and Trade, Dongguan City College, Dongguan, China, 523419}

%\author[a]{Xiaoliang Li}
%
%\author[b]{Author 2\thanks{Corresponding author: name@mail.com}}
%
%\affil[a]{School of Finance and Trade, Dongguan City College, Dongguan, China, 523419}
%
%\affil[b]{School, University, City, Country}

\maketitle

\begin{abstract}
In this discussion draft, we investigate five different models of duopoly games, where the market is assumed to have an isoelastic demand function. Moreover, quadratic cost functions reflecting decreasing returns to scale are considered. The games in this draft are formulated with systems of two nonlinear difference equations. Existing equilibria and their local stability are analyzed by symbolic computations. In the model where a gradiently adjusting player and a rational (or a boundedly rational) player compete with each other, diseconomies of scale are proved to have an effect of stability enhancement, which is consistent with the similar results found by Fisher for homogeneous oligopolies with linear demand functions.
\end{abstract}

%\section{Introduction}

\section{Models}

Let us consider a market served by two firms producing homogeneous products. We use $q_i(t)$ to denote the output of firm $i$ at period $t$. Moreover, the cost function of firm $i$ is supposed to be quadratic, i.e., $C_i(q_i)=c_iq_i^2$ with $c_i>0$. At each period $t$, firm $i$ first estimates the possible price $p_i^e(t)$ of the product, then the expected profit of firm $i$ would be 
$$\Pi_i^e(t)=p_i^e(t) q_i(t) - c_iq_i^2(t),~~i=1,2.$$
In order to maximize the expected profit,  at period $t$ each firm would decide the quantity of the output by solving
$$q_i(t)=\arg\max_{q_i(t)}\Pi_i^e(t)=\arg\max_{q_i(t)}\left[p_i^e(t) q_i(t) - c_iq_i^2(t)\right],~~i=1,2.$$

Furthermore, assume that the demand function of the market is isoelastic, which is founded on the hypothesis that the consumers have the Cobb-Douglas utility function. Hence, the real (not expected) price of the product should be
$$p(Q)=\frac{1}{Q}=\frac{1}{q_1+q_2},$$
where $Q=q_1+q_2$ is the total supply. Five types of players with distinct rationality degrees are involved in this draft, which are described in detail as follows.

A \emph{rational player} not only knows clearly the form of the price function, but also has complete information of the decision of its rival. If firm $i$ is a rational player, at period $t+1$ we have
$$p_i^e(t+1)=\frac{1}{q_i(t+1)+q_{-i}^e(t+1)},$$
where $q_{-i}^e(t+1)$ is the expectation of the output of the rival. Due to the assumption of complete information, which means that $q_{-i}^e(t+1)=q_{-i}(t+1)$, it is acquired that the expected profit of firm $i$ would be
$$\Pi_i^e(t+1)=\frac{q_i(t+1)}{q_i(t+1)+q_{-i}(t+1)}-c_iq_i^2(t+1)$$
The first order condition for profit maximization gives rise to a cubic polynomial equation. To be exact, the condition for the reaction function of firm $i$ would be
\begin{equation}\label{eq:rational-cd-r}
q_{-i}(t+1)-2\,c_iq_i(t+1)(q_i(t+1)+q_{-i}(t+1))^2=0,
\end{equation}
which is simply denoted as $F_i(q_i(t+1),q_{-i}(t+1))= 0$ in the sequel.
The player could maximize its profit by solving the above equation. It is easy to verify that there exists only one positive solution for $q_i(t+1)$ if solving \eqref{eq:rational-cd-r}, but the expression could be quite complex. 
\begin{equation}\label{eq:r_i}
	q_i(t+1)=\frac{\sqrt[3]{2}M}{6c_i}+\frac{\sqrt[3]{4}c_iq_{-i}^2(t+1)}{3M}-\frac{2q_{-i}(t+1)}{3},
\end{equation}
where
$$M=\sqrt[3]{c_i^2	q_{-i}(t+1)(4c_iq_{-i}^2(t+1)+3\sqrt{3}\sqrt{8c_iq_{-i}^2(t+1)+27}+27)}$$

For simplicity, we temporarily denote \eqref{eq:r_i} by 
$$q_i(t+1)=R_i(q_{-i}(t+1)),$$ 
where $R_i$ is called the \emph{reaction function} of firm $i$. It is evident that if the two firms in the market are both rational players, the equilibrium (the best decision of output) would be arrived in a shot and there are no dynamics in the system. In order to tackle this problem, Puu introduced the bounded rational player in \cite{Puu1991C}.

A \emph{boundedly rational} player knows the form of the price function, but do not know the rival's decision of the production. If firm $i$ is a boundedly rational player, then it naively expects its competitor to produce the same quantity of output as the last period, i.e., $q_{-i}^e(t+1)=q_{-i}(t)$. Thus,
$$\Pi_i^e(t+1)=\frac{q_i(t+1)}{q_i(t+1)+q_{-i}(t)}-c_iq_i^2(t+1).$$
Then the best response for firm $i$ would be $q_i(t+1)=R_i(q_{-i}(t))$. 

A \emph{local monopolistic approximation} (LMA) player, which even does not know the exact form the price function, is less rational than a boundedly rational player. Specifically, if firm $i$ is an LMA player, then it just can  observe the current market price $p(t)$ and the corresponding total supply $Q(t)$, and is able to correctly estimate the slope $p'(Q(t))$ of the price function around the point $(p(t),Q(t))$.  Then firm $i$ uses such information to conjecture the demand function and expect the price at period $t+1$ to be
$$p_i^e(t+1)=p(Q(t))+p'(Q(t))(Q_i^e(t+1)-Q(t)),$$
where $Q_i^e(t+1)=q_i(t+1)+q_{-i}^e(t+1)$ represents the expected aggregate production of firm $i$ at period $t+1$. Moreover, an LMA player do not know the decision of its rival either, and is assumed to to use the naive expectation, i.e., $q_{-i}^e(t+1)=q_{-i}(t)$. Thus,  
$$p_i^e(t+1)=\frac{1}{Q(t)}-\frac{1}{Q^2(t)}(q_i(t+1)-q_i(t)).$$
The expected profit would be
$$\Pi^e_i(t+1)=q_i(t+1)\left[\frac{1}{Q(t)}-\frac{1}{Q^2(t)}(q_i(t+1)-q_i(t))\right]-c_iq_i^2(t+1).$$
By solving the first order condition, the best response for firm $i$ would be
$$q_i(t+1)=\frac{2\,q_i(t)+q_{-i}(t)}{2(1+c_i(q_i(t)+q_{-i}(t))^2)}.$$
For simplicity, we denote the above map as $q_i(t+1)=S_i(q_i(t),q_{-i}(t))$. 

In addition, we have an \emph{adaptive} player that decides the quantity of production according to its output of the previous period as well as the expectation of its rival. Specifically, if firm $i$ is an adaptive player, then at period $t+1$ this firm naively expects its competitor would produce the same quantity of output as the last period, i.e., $q_{-i}^e(t+1)=q_{-i}(t)$. Then the best response for firm $i$ would be $q_i(t+1)=R_i(q_{-i}(t))$. The adaptive decision mechanism for firm $i$ is that it choose the output $q_i(t+1)$ proportionally to be
$$q_i(t+1)=(1-L)q_i(t)+LR_i(q_{-i}(t)),$$
where $L\in(0,1)$ is a parameter controlling the proportion. It should be noticed that an adaptive player degenerate into a boundedly player if we suppose $L=1$. 

Furthermore, we consider a \emph{gradiently adjusting} player, which increases/decreases its output according to the information given by the marginal profit of the last period. Specifically, if firm $i$ is a gradiently adjusting player, then at period $t+1$ this firm is supposed to know its own profit function at period $t$, that is
$$\Pi_i(t)=\frac{q_i(t)}{q_i(t)+q_{-i}(t)}-c_iq_i^2(t).$$
Hence, firm $i$ could adjust its output at period $t+1$ with a gradient mechanism as
\begin{equation}
	q_i(t+1)=q_i(t) + K q_i(t) \frac{\partial \Pi_i(t)}{\partial q_i(t)},
\end{equation}
where 
$$\frac{\partial \Pi_i(t)}{\partial q_i(t)}=\frac{q_{-i}(t)}{(q_i(t)+q_{-i}(t))^2}-2c_iq_i(t)$$ 
is the marginal profit of firm $i$ as period $t$, and $K>0$ is a parameter controlling the adjustment speed. It is worth noting that the adjustment speed depends upon not only the parameter $K$ but also the size of the firm $q_i(t)$. One may observe from the above iteration map that a gradient adjusting player does not need to expect or guess the production output of its rival at the current period. Denote 
$$G_i(q_i(t),q_{-i}(t))=q_i(t)\frac{\partial \Pi_i(t)}{\partial q_i(t)}=\frac{q_i(t)q_{-i}(t)}{(q_i(t)+q_{-i}(t))^2}-2c_iq_i^2(t).$$ 
Naturally, the following models could be considered.

\begin{model}[GR]
\begin{equation}\label{eq:sys-gr}
	M_{GR}(q_1,q_2): 
	\left\{\begin{split}
		&q_1(t+1)=q_1(t)+K G_1(q_1(t),q_2(t)),\\
		&q_2(t+1)=R_2(q_1(t+1)),
	\end{split}
	\right.
\end{equation}
where $K>0$.
\end{model}

\begin{model}[GB]
	\begin{equation}\label{eq:sys-gr}
	M_{GB}(q_1,q_2): 
	\left\{\begin{split}
		&q_1(t+1)=q_1(t)+K G_1(q_1(t),q_2(t)),\\
		&q_2(t+1)=R_2(q_1(t)),
	\end{split}
	\right.
\end{equation}
where $K>0$.
\end{model}

\begin{model}[GL]
\begin{equation}\label{eq:sys-gl}
	M_{GL}(q_1,q_2): 
	\left\{\begin{split}
		&q_1(t+1)=q_1(t)+K G_1(q_1(t),q_2(t)),\\
		&q_2(t+1)=S_2(q_1(t),q_2(t)),
	\end{split}
	\right.
\end{equation}	
where $K>0$.
\end{model}

\begin{model}[GA]
\begin{equation}\label{eq:sys-gl}
	M_{GL}(q_1,q_2): 
	\left\{\begin{split}
		&q_1(t+1)=q_1(t)+K G_1(q_1(t),q_2(t)),\\
		&q_2(t+1)=(1-L)q_2(t)+ LR_2(q_1(t)),
	\end{split}
	\right.
\end{equation}	
where $K>0$ and $0<L<1$.
\end{model}

\begin{model}[GG]
	\begin{equation}\label{eq:sys-gg}
	M_{GG}(q_1,q_2): 
	\left\{\begin{split}
		&q_1(t+1)=q_1(t)+K_1 G_1(q_1(t),q_2(t)),\\
		&q_2(t+1)=q_2(t)+K_2 G_2(q_2(t),q_1(t)),\\
	\end{split}
	\right.
\end{equation}
where $K_1>0$ and $K_2>0$.
\end{model}

It is easy to verify that all the above models has one unique Nash equilibrium, of which the closed-form expression is 
$$E=\left[~\frac{\sqrt{c_2}}{\sqrt{c_1}+\sqrt{c_2}}\frac{1}{\sqrt{2\sqrt{c_1c_2}}},~ \frac{\sqrt{c_1}}{\sqrt{c_1}+\sqrt{c_2}}\frac{1}{\sqrt{2\sqrt{c_1c_2}}}~\right].$$

\section{Local Stability}

In this section, we study the local stability of the unique Nash equilibrium of all the models. Indeed, for each model, this problem could be transformed into determining the existence of real solutions of a system formulated by polynomial equations and inequalities. Afterward, the symbolic approach proposed by the author and his coworker in \cite{Li2014C} is used to systematically address the resulting systems. It should be noticed that the processes of computations in this paper are similar to \cite{Li2021D}, but the considered models are different.

\subsection{Model GR}

This model could be equivalently described by a one-dimensional iteration map as follows.
\begin{equation}\label{eq:gr-map-dim1}
	M_{GR}(q_1): 
	q_1(t+1)=q_1(t)+Kq_1(t)G_1(q_1(t),R_2(q_1(t))).
\end{equation}
Hence, the stable equilibria are the solutions of the following system.
\begin{equation}\label{eq:gr-system}
	\left\{\begin{split}
		&K\cdot G_1(q_1,q_2)=0,\\
		&F_2(q_2,q_1)=0,\\
		&q_1>0,~q_2>0,\\
		&1+ \left(1+K\frac{{\rm d}[q_1G_1(q_1,R_2(q_1))]}{{\rm d} q_1}\right)>0,\\
		&1- \left(1+K\frac{{\rm d}[q_1G_1(q_1,R_2(q_1))]}{{\rm d} q_1}\right)>0,\\
		&c_1>0,~c_2>0,~K>0.
	\end{split}
	\right.
\end{equation}

For the above system, the squarefree part of the border polynomial is
\begin{align*}
	SP^*_{GR}=\,&c_1c_2K(c_1-c_2) (c_1+c_2) (c_1-1/9\,c_2)(c_1^3c_2K^4-3/2\,c_1^2c_2K^2-81/64\,c_1^2+9/32\,c_1c_2-1/64\,c_2^2).
\end{align*}
%$$SP^*_{GR}=ca*cb*k*(ca-cb)*(ca+cb)*(ca-1/9*cb)*(ca^3*cb*k^4-3/2*ca^2*cb*k^2-81/64*ca^2+9/32*ca*cb-1/64*cb^2).$$
We select the sample points as
$$(1, 1/2, 2),~(1, 2, 1),~(1, 2, 2),~(1, 10, 1),~(1, 10, 2).$$

By checking the number of real solution of \eqref{eq:gr-system} at these sample points, the following results are finally acquired.

\begin{theorem}
For Model GR, there exists a unique equilibrium with $q_1,q_2>0$. Moreover, this equilibrium is locally stable if $R_{GR}^1<0$, where 
$$R_{GR}^1=
64\,c_1^3c_2K^4-96\,c_1^2c_2K^2-81\,c_1^2+18\,c_1c_2-c_2^2.
$$

\end{theorem}

%`Critical polynomial is:`
%`LXL cp=`, ca*cb*k*(ca+cb)*(ca*k+cb*k-4)
%
%`LXL sample:`, [1, 1, 1], `num:`, 1, `time:`, .6e-2
%`LXL sample:`, [1, 1, 3], `num:`, 0, `time:`, .3e-2

If we consider the counterpart of this model with quadratic costs replaced by linear costs, the following theorem is obtained.

\begin{theorem}
For Model GR, if $C_1(q_1)=c_1q_1$ and $C_2(q_2)=c_2q_2$, there exists a unique equilibrium with $q_1,q_2>0$. Moreover, this equilibrium is locally stable if $R_{GR}^2<0$, where
$$R_{GR}^2=c_1K+c_2K-4<0.$$
%$$R_{GR}^2=ca*k+cb*k-4<0.$$
\end{theorem}

\begin{figure}[htbp]
  \centering
    \subfigure[]{\includegraphics[width=0.4\textwidth]{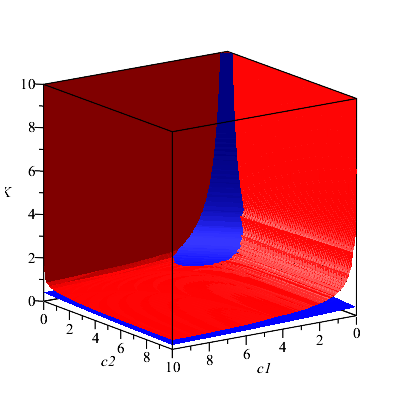}} 
    \subfigure[]{\includegraphics[width=0.4\textwidth]{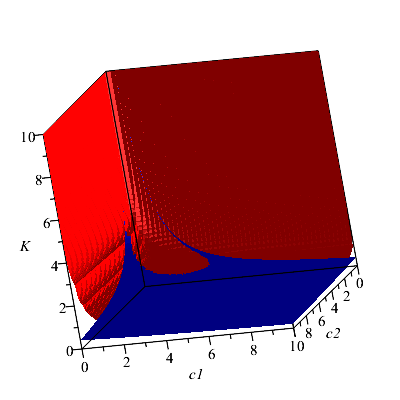}}\\
    \subfigure[$c_1=c_2$.]{\includegraphics[width=0.4\textwidth]{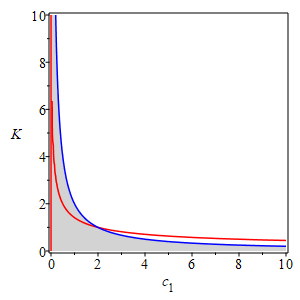}} 
    \subfigure[$K=1$.]{\includegraphics[width=0.4\textwidth]{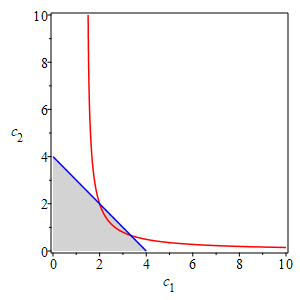}}
  \caption{The 3-dimensional $(c_1,c_2,K)$ parameter space of Model GR. The red surface is $R_{GR}^1=0$, and the blue surface is $R_{GR}^2=0$.}
    \label{fig:par-gr}
\end{figure}

The following proposition is consistent with the stability enhancement effect of diseconomies of scale found by Fisher in \cite{Fisher1961T}.

\begin{proposition}
	For Model GR, if $c_1>4$ or $c_2>3$, the stable region for the linear costs $C_i(q_i)=c_iq_i$ is strictly contained in that for the quadratic costs $C_i(q_i)=c_iq_i^2$.
\end{proposition}
\begin{proof}
The proof is tedious and we leave it to the readers.
\end{proof}

\subsection{Model GB}

The Jacobian matrix is
\begin{equation}
J_{GB}=\left[
	\begin{matrix}
	1+K\cdot\partial G_1/\partial q_1 & K\cdot\partial G_1/\partial q_2\\
	{\rm d} R_2/{\rm d} q_1 & 0
	\end{matrix}
\right]
\end{equation}
Hence, the stable equilibria can be described by 
\begin{equation}\label{eq:semi-lb}
	\left\{\begin{split}
		&K\cdot G_1(q_1,q_2)=0,\\
		&F_2(q_2,q_1)=0,\\
		&q_1>0,~q_2>0,\\
		&1+\Tr(J_{GB})+\Det(J_{GB})>0,\\
		&1-\Tr(J_{GB})+\Det(J_{GB})>0,\\
		&1-\Det(J_{GB})>0,\\
		&c_1>0,~c_2>0,~K>0.
	\end{split}
	\right.
\end{equation}

%ca*cb*k*(ca+cb)*(ca-cb)*(ca-1/9*cb)*(ca^7*cb*k^4-4*ca^6*cb^2*k^4+6*ca^5*cb^3*k^4-4*ca^4*cb^4*k^4+ca^3*cb^5*k^4-78*ca^6*cb*k^2-680*ca^5*cb^2*k^2-980*ca^4*cb^3*k^2-296*ca
%^3*cb^4*k^2-14*ca^2*cb^5*k^2-81/4*ca^6+171/2*ca^5*cb-559/4*ca^4*cb^2+109*ca^3*cb^3-159/4*cb^4*ca^2+11/2*cb^5*ca-1/4*cb^6)*(ca^7*cb*k^4-68*ca^6*cb^2*k^4+1158*ca^5*cb^3*k^4-68*ca^4*
%cb^4*k^4+ca^3*cb^5*k^4+66*ca^6*cb*k^2-616*ca^5*cb^2*k^2-1524*ca^4*cb^3*k^2+24*ca^3*cb^4*k^2+2*ca^2*cb^5*k^2-81/4*ca^6+171/2*ca^5*cb-559/4*ca^4*cb^2+109*ca^3*cb^3-159/4*cb^4*ca^2+
%11/2*cb^5*ca-1/4*cb^6)
%
%[[[1, 1/64, 1], 1], [[1, 1/2, 1], 1], [[1, 2, 1], 1], [[1, 10, 1], 1], [[1, 34, 1], 1]], [[[1, 1/64, 1], 1], [[1, 1/2, 1], 1], [[1, 2, 1],
%1], [[1, 10, 1], 1], [[1, 34, 1], 1]]

Afterward, we acquire the following theorem.

\begin{theorem}
For Model GB, there exists a unique equilibrium with $q_1,q_2>0$. Moreover, this equilibrium is locally stable if $R_{GB}^1<0$,
where 
\begin{align*}
R_{GB}^1=&\,
4\,c_1^7c_2K^4-272\,c_1^6c_2^2K^4+4632\,c_1^5c_2^3K^4-272\,c_1^4c_2^4K^4+4\,c_1^3c_2^5K^4+264\,c_1^6c_2K^2\\
&-2464\,c_1^5c_2^2K^2-6096\,c_1^4c_2^3K^2+96\,c_1^3c_2^4K^2+8\,c_1^2c_2^5K^2-81\,c_1^6+342\,c_1^5c_2\\
&-559\,c_1^4c_2^2+436\,c_1^3c_2^3-159\,c_1^2c_2^4+22\,c_1c_2^5-c_2^6.
\end{align*}
%$$R_{GB}^1=
%4*ca^7*cb*k^4-272*ca^6*cb^2*k^4+4632*ca^5*cb^3*k^4-272*ca^4*cb^4*k^4+4*ca^3*cb^5*k^4+264*ca^6*cb*k^2-2464*ca^5*cb^2*k^2-6096*ca^4*cb^3*k^2+96*ca^3*cb^4*k^2+8*ca^2*cb^5*k^2-81*ca^6+342*ca^5*cb-559*ca^4*cb^2+436*ca^3*cb^3-159*ca^2*cb^4+22*ca*cb^5-cb^6.
%$$
\end{theorem}

%ca*cb*k*(ca+cb)*(ca^2*k-6*ca*cb*k+cb^2*k+4*ca+4*cb)*(ca^2*k-2*ca*cb*k+cb^2*k-2*ca-2*cb)*(ca^2*k-2*ca*cb*k+cb^2*k+2*ca+2*cb)
%
% [[[1, 1/8, 1], 1], [[1, 1/4, 1], 1], [[1, 1/2, 1], 1], [[1, 2, 1], 1], [[1, 4, 1], 1], [[1, 6, 1/2], 1]], [[[1, 1/8, 1], 1], [[1, 1/4, 1],
%1], [[1, 1/2, 1], 1], [[1, 2, 1], 1], [[1, 4, 1], 1], [[1, 6, 1/2], 1]]

The linear case was first studied in \cite{Tramontana2010H} and is restated as follows.

\begin{proposition}
For Model GB, if $C_1(q_1)=c_1q_1$ and $C_2(q_2)=c_2q_2$, there exists a unique equilibrium with $q_1,q_2>0$. Moreover, this equilibrium is locally stable if $R_{GB}^2>0$ and $R_{GB}^3<0$, where
\begin{align*}
	R_{GB}^2=\,& c_1^2K-6\,c_1c_2K+c_2^2K+4\,c_1+4\,c_2,\\
	R_{GB}^2=\,& c_1^2K-2\,c_1c_2K+c_2^2K-2\,c_1-2\,c_2.
\end{align*}
%\begin{align*}
%	R_{GB}^2=& ca^2*k-6*ca*cb*k+cb^2*k+4*ca+4*cb,\\
%	R_{GB}^2=& ca^2*k-2*ca*cb*k+cb^2*k-2*ca-2*cb.
%\end{align*}
\end{proposition}

%$k>4(c_1+c_2)/(c_1^2-6c_1c_2+c_2^2)$
%
%$k<2(c_1+c_2)/(c_1-c_2)^2$

\begin{figure}[htbp]
  \centering
    \subfigure[]{\includegraphics[width=0.4\textwidth]{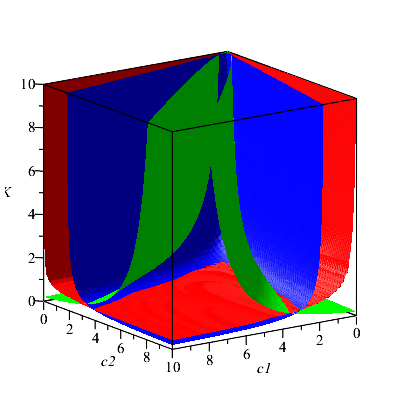}} 
    \subfigure[]{\includegraphics[width=0.4\textwidth]{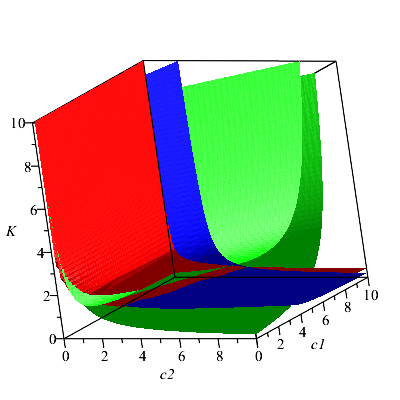}} \\
    \subfigure[$c_1=c_2$.]{\includegraphics[width=0.4\textwidth]{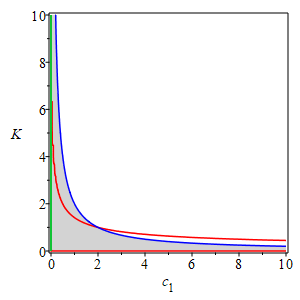}} 
    \subfigure[$K=1$.]{\includegraphics[width=0.4\textwidth]{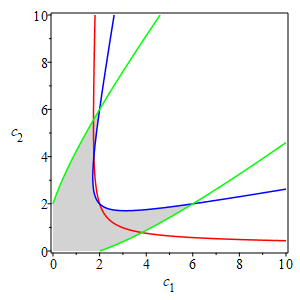}} \\
  \caption{The 3-dimensional $(c_1,c_2,K)$ parameter space of Model GB. The red surface is $R_{GB}^1=0$, the blue surface is $R_{GB}^2=0$, and the green surface is $R_{GB}^3=0$.}
    \label{fig:par-gb}
\end{figure}

The following result is similar to the paper by Fisher \cite{Fisher1961T}.

\begin{proposition}
	For Model GB, if $c_1>13$ or $c_2>7$, the stable region for the linear costs $C_i(q_i)=c_iq_i$ is strictly contained in that for the quadratic costs $C_i(q_i)=c_iq_i^2$.
\end{proposition}

\begin{proof}
The proof is tedious and we leave it to the readers.
\end{proof}

\subsection{Model GL}

The Jacobian matrix is
\begin{equation}
J_{GL}=\left[
	\begin{matrix}
	1+K\cdot \partial G_1/\partial q_1 & K\cdot \partial G_1/\partial q_2\\
	\partial S_2/\partial q_1 & \partial S_2/\partial q_2
	\end{matrix}
\right]
\end{equation}

Hence, the stable equilibria can be described by 
\begin{equation}\label{eq:semi-gl}
	\left\{\begin{split}
		&K\cdot G_1(q_1,q_2)=0,\\
		&S_2(q_2,q_1)=0,\\
		&q_1>0,~q_2>0,\\
		&1+\Tr(J_{GL})+\Det(J_{GL})>0,\\
		&1-\Tr(J_{GL})+\Det(J_{GL})>0,\\
		&1-\Det(J_{GL})>0,\\
		&c_1>0,~c_2>0,~K>0.
	\end{split}
	\right.
\end{equation}

% ca*cb*k*(ca-cb)*(ca+cb)*(-1/4*cb+ca)*(ca^7*k^4-13/2*ca^6*cb*k^4+177/16*ca^5*cb^2*k^4-13/8*ca^4*cb^3*k^4+1/16*ca^3*cb^4*k^4+3/4*ca^5*cb*k^2+29/8*ca^4*cb^2*k^2-23/2*ca^3*
%cb^3*k^2-7/8*ca^2*cb^4*k^2-81/64*cb*ca^4+45/16*ca^3*cb^2-59/32*ca^2*cb^3+5/16*ca*cb^4-1/64*cb^5)*(ca^7*cb*k^4-13/2*ca^6*cb^2*k^4+177/16*ca^5*cb^3*k^4-13/8*ca^4*cb^4*k^4+1/16*ca^3*
%cb^5*k^4-2*ca^6*cb*k^2+85/4*ca^5*cb^2*k^2-267/8*ca^4*cb^3*k^2-17*ca^3*cb^4*k^2-7/8*ca^2*cb^5*k^2-4*ca^6+17/2*cb*ca^5-353/64*ca^4*cb^2+25/16*ca^3*cb^3-19/32*ca^2*cb^4+1/16*ca*cb^5-\
%1/64*cb^6)*(ca^7*cb*k^4-21/2*ca^6*cb^2*k^4+449/16*ca^5*cb^3*k^4-21/8*ca^4*cb^4*k^4+1/16*ca^3*cb^5*k^4+6*ca^6*cb*k^2-25/4*ca^5*cb^2*k^2-267/8*ca^4*cb^3*k^2+3/2*ca^3*cb^4*k^2+1/8*ca
%^2*cb^5*k^2-4*ca^6+17/2*cb*ca^5-353/64*ca^4*cb^2+25/16*ca^3*cb^3-19/32*ca^2*cb^4+1/16*ca*cb^5-1/64*cb^6)
%
%`all parametric points verified`, 3.495, [[[1, 1/8, 1/2], 1], [[1, 1/8, 1], 1], [[1, 7/32, 1/2], 1], [[1, 7/32, 1], 1], [[1, 17/64, 1/2], 1], [[1, 17/64, 1], 1], [[1, 9/32, 1/2],
%1], [[1, 9/32, 1], 1], [[1, 1/2, 1], 1], [[1, 1/2, 3/2], 1], [[1, 2, 1], 1], [[1, 5, 1], 1], [[1, 10, 1], 1], [[1, 13, 1], 1], [[1, 21, 1], 1]], [[[1, 1/8, 1/2], 1], [[1, 1/8, 1],
%1], [[1, 7/32, 1/2], 1], [[1, 7/32, 1], 1], [[1, 17/64, 1/2], 1], [[1, 17/64, 1], 1], [[1, 9/32, 1/2], 1], [[1, 9/32, 1], 1], [[1, 1/2, 1], 1], [[1, 1/2, 3/2], 1], [[1, 2, 1], 1],
%[[1, 5, 1], 1], [[1, 10, 1], 1], [[1, 13, 1], 1], [[1, 21, 1], 1]]

\begin{theorem}
For Model GL, there exists a unique equilibrium with $q_1,q_2>0$. Moreover, this equilibrium is locally stable if $R_{GL}<0$,
where 
\begin{align*}
	R_{GL}^1=\,&
64\,c_1^7c_2K^4-672\,c_1^6c_2^2K^4+1796\,c_1^5c_2^3K^4-168\,c_1^4c_2^4K^4+4\,c_1^3c_2^5K^4+384\,c_1^6c_2K^2\\
&-400\,c_1^5c_2^2K^2-2136\,c_1^4c_2^3K^2+96\,c_1^3c_2^4K^2+8\,c_1^2c_2^5K^2-256\,c_1^6+544\,c_1^5c_2\\
&-353\,c_1^4c_2^2+100\,c_1^3c_2^3-38\,c_1^2c_2^4+4\,c_1c_2^5-c_2^6.
\end{align*}
%$$R_{GL}^1=
%64*ca^7*cb*k^4-672*ca^6*cb^2*k^4+1796*ca^5*cb^3*k^4-168*ca^4*cb^4*k^4+4*ca^3*cb^5*k^4+384*ca^6*cb*k^2-400*ca^5*cb^2*k^2-2136*ca^4*cb^3*k^2+96*ca^3*cb^4*k^2+8*ca^2*cb^5*k^2-256*ca^6+544*ca^5*cb-353*ca^4*cb^2+100*ca^3*cb^3-38*ca^2*cb^4+4*ca*cb^5-cb^6.
%$$
\end{theorem}

%ca*cb*k*(ca+cb)*(ca*k-1/3*cb*k+2/3)*(ca*cb*k-1/7*cb^2*k-8/7*ca-4/7*cb)*(ca*cb*k-1/3*cb^2*k-4/3*ca-2/3*cb)
%
%[[[1, 1, 1], 1], [[1, 7/2, 1], 1], [[1, 5, 1/2], 1], [[1, 8, 1/4], 1]], [[[1, 1, 1], 1], [[1, 7/2, 1], 1], [[1, 5, 1/2], 1], [[1, 8, 1/4],
%1]]

The linear case was first studied in \cite{Cavalli2015N}, which is restated here.

\begin{proposition}
For Model GL, if $C_1(q_1)=c_1q_1$ and $C_2(q_2)=c_2q_2$, there exists a unique equilibrium with $q_1,q_2>0$. Moreover, this equilibrium is locally stable if $R_{GL}^2>0$ and $R_{GL}^3<0$, where
\begin{align*}
	R_{GL}^2 =&\, 3\,c_1K-c_2K+2,\\
	R_{GL}^3 =&\, 7\,c_1c_2K-c_2^2K-8\,c_1-4\,c_2.
\end{align*}
%\begin{align*}
%	R_{GL}^2 =& 3*ca*k-cb*k+2,\\
%	R_{GL}^3 =& 7*ca*cb*k-cb^2*k-8*ca-4*cb.
%\end{align*}
\end{proposition}

\begin{figure}[htbp]
  \centering
    \subfigure[]{\includegraphics[width=0.4\textwidth]{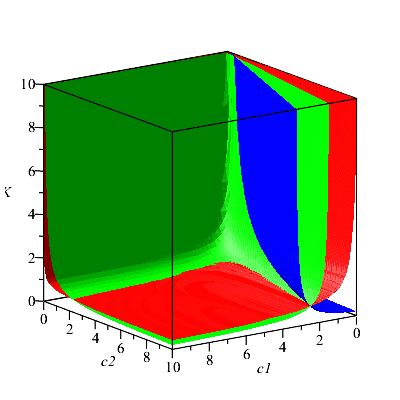}} 
    \subfigure[]{\includegraphics[width=0.4\textwidth]{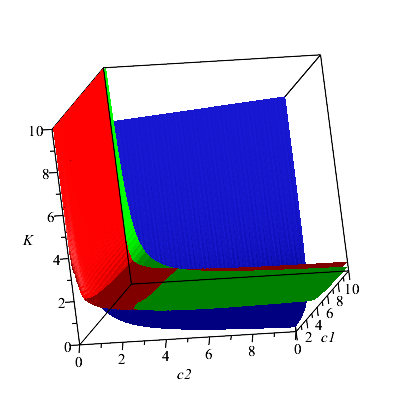}} \\
    
    \subfigure[$c_1=c_2$.]{\includegraphics[width=0.4\textwidth]{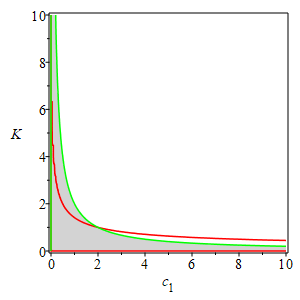}} 
    \subfigure[$K=1$.]{\includegraphics[width=0.4\textwidth]{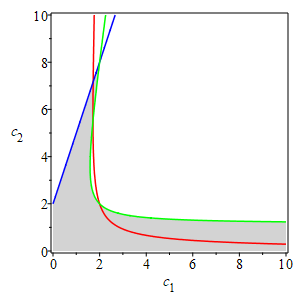}} \\
    
  \caption{The 3-dimensional $(c_1,c_2,K)$ parameter space of Model GL. The red surface is $R_{GL}^1=0$, the blue surface is $R_{GL}^2=0$, and the green surface is $R_{GL}^3=0$.}
    \label{fig:par-gl}
\end{figure}

\begin{proposition}
	For Model GL, even if $c_1>10^{100}$ and $c_2>10^{100}$, the stable region for the linear costs $C_i(q_i)=c_iq_i$ is not strictly contained in that for the quadratic costs $C_i(q_i)=c_iq_i^2$.
\end{proposition}

\subsection{Model GA}

The Jacobian matrix is
\begin{equation}
J_{GA}=\left[
	\begin{matrix}
	1+K\cdot\partial G_1/\partial q_1 & K\cdot\partial G_1/\partial q_2\\
	L\cdot {\rm d} R_2/{\rm d} q_1 & 1-L	\end{matrix}
\right]
\end{equation}

Hence, the stable equilibria can be described by 
\begin{equation}\label{eq:semi-gg}
	\left\{\begin{split}
		&K\cdot G_1(q_1,q_2)=0,\\
		&F_2(q_2,q_1)=0,\\
		&q_1>0,~q_2>0,\\
		&1+\Tr(J_{GA})+\Det(J_{GA})>0,\\
		&1-\Tr(J_{GA})+\Det(J_{GA})>0,\\
		&1-\Det(J_{GA})>0,\\
		&c_1>0,~c_2>0,~K>0,~L>0,~1-L>0.
	\end{split}
	\right.
\end{equation}

% ca*cb*ka*kb*(ca-cb)*(ca+cb)*(-1024*ca^3*cb^3*ka^4*kb^4+384*ca^3*cb^2*ka^4*kb^2+384*ca^3*cb^2*ka^3*kb^3+384*ca^2*cb^3*ka^3*kb^3+384*ca^2*cb^3*ka^2*kb^4+ca^4*ka^4-18*ca^3
%*cb*ka^4-32*ca^3*cb*ka^3*kb-18*ca^3*cb*ka^2*kb^2+81*ca^2*cb^2*ka^4+288*ca^2*cb^2*ka^3*kb+420*ca^2*cb^2*ka^2*kb^2+288*ca^2*cb^2*ka*kb^3+81*ca^2*cb^2*kb^4-18*ca*cb^3*ka^2*kb^2-32*ca
%*cb^3*ka*kb^3-18*ca*cb^3*kb^4+cb^4*kb^4)*(-64*ca^4*cb^4*ka^4*kb^4+96*ca^4*cb^3*ka^4*kb^2-32*ca^4*cb^3*ka^3*kb^3-32*ca^3*cb^4*ka^3*kb^3+96*ca^3*cb^4*ka^2*kb^4+ca^5*cb*ka^4-18*ca^4*
%cb^2*ka^4+32*ca^4*cb^2*ka^3*kb-18*ca^4*cb^2*ka^2*kb^2+81*ca^3*cb^3*ka^4+96*ca^3*cb^3*ka^3*kb-92*ca^3*cb^3*ka^2*kb^2+96*ca^3*cb^3*ka*kb^3+81*ca^3*cb^3*kb^4-18*ca^2*cb^4*ka^2*kb^2+
%32*ca^2*cb^4*ka*kb^3-18*ca^2*cb^4*kb^4+ca*cb^5*kb^4+8*ca^4*cb*ka^2-8*ca^4*cb*ka*kb-16*ca^3*cb^2*ka^2-120*ca^3*cb^2*ka*kb-120*ca^3*cb^2*kb^2-120*ca^2*cb^3*ka^2-120*ca^2*cb^3*ka*kb-\
%16*ca^2*cb^3*kb^2-8*ca*cb^4*ka*kb+8*ca*cb^4*kb^2-4*ca^4+16*ca^3*cb-24*ca^2*cb^2+16*ca*cb^3-4*cb^4)*(-1024*ca^4*cb^4*ka^4*kb^4+384*ca^4*cb^3*ka^4*kb^2-640*ca^4*cb^3*ka^3*kb^3-640*
%ca^3*cb^4*ka^3*kb^3+384*ca^3*cb^4*ka^2*kb^4+ca^5*cb*ka^4-18*ca^4*cb^2*ka^4+96*ca^4*cb^2*ka^3*kb-210*ca^4*cb^2*ka^2*kb^2+81*ca^3*cb^3*ka^4-96*ca^3*cb^3*ka^3*kb-732*ca^3*cb^3*ka^2*
%kb^2-96*ca^3*cb^3*ka*kb^3+81*ca^3*cb^3*kb^4-210*ca^2*cb^4*ka^2*kb^2+96*ca^2*cb^4*ka*kb^3-18*ca^2*cb^4*kb^4+ca*cb^5*kb^4+8*ca^4*cb*ka^2-40*ca^4*cb*ka*kb-16*ca^3*cb^2*ka^2-88*ca^3*
%cb^2*ka*kb-120*ca^3*cb^2*kb^2-120*ca^2*cb^3*ka^2-88*ca^2*cb^3*ka*kb-16*ca^2*cb^3*kb^2-40*ca*cb^4*ka*kb+8*ca*cb^4*kb^2-4*ca^4+16*ca^3*cb-24*ca^2*cb^2+16*ca*cb^3-4*cb^4)

\begin{theorem}
For Model GG, there exists a unique equilibrium with $q_1,q_2>0$. Moreover, this equilibrium is locally stable if $R_{GG}^1<0$,
where
\begin{align*} 
\begin{autobreak}
R_{GA}^1  = 
64\,c_1^{7} c_2 K^{4} L^{4} 
- 256\,c_1^{6} c_2^{2} K^{4} L^{4} 
+ 384\,c_1^{5} c_2^{3} K^{4} L^{4} 
- 256\,c_1^{4} c_2^{4} K^{4} L^{4} 
+ 64\,c_1^{3} c_2^{5} K^{4} L^{4} 
- 384\,c_1^{7} c_2 K^{4} L^{3} 
+ 2560\,c_1^{6} c_2^{2} K^{4} L^{3} 
- 4352\,c_1^{5} c_2^{3} K^{4} L^{3} 
+ 2560\,c_1^{4} c_2^{4} K^{4} L^{3} 
- 384\,c_1^{3} c_2^{5} K^{4} L^{3} 
+ 864\,c_1^{7} c_2 K^{4} L^{2} 
- 8576\,c_1^{6} c_2^{2} K^{4} L^{2} 
+ 19520\,c_1^{5} c_2^{3} K^{4} L^{2} 
- 8576\,c_1^{4} c_2^{4} K^{4} L^{2} 
+ 864\,c_1^{3} c_2^{5} K^{4} L^{2} 
- 864\,c_1^{7} c_2 K^{4} L 
+ 11904\,c_1^{6} c_2^{2} K^{4} L 
- 96\,c_1^{6} c_2 K^{2} L^{4} 
- 38464\,c_1^{5} c_2^{3} K^{4} L 
+ 384\,c_1^{5} c_2^{2} K^{2} L^{4} 
+ 11904\,c_1^{4} c_2^{4} K^{4} L 
- 576\,c_1^{4} c_2^{3} K^{2} L^{4} 
- 864\,c_1^{3} c_2^{5} K^{4} L 
+ 384\,c_1^{3} c_2^{4} K^{2} L^{4} 
- 96\,c_1^{2} c_2^{5} K^{2} L^{4} 
+ 324\,c_1^{7} c_2 K^{4} 
- 5904\,c_1^{6} c_2^{2} K^{4} 
+ 96\,c_1^{6} c_2 K^{2} L^{3} 
+ 27544\,c_1^{5} c_2^{3} K^{4} 
- 2944\,c_1^{5} c_2^{2} K^{2} L^{3} 
- 5904\,c_1^{4} c_2^{4} K^{4} 
+ 6208\,c_1^{4} c_2^{3} K^{2} L^{3} 
+ 324\,c_1^{3} c_2^{5} K^{4} 
- 3968\,c_1^{3} c_2^{4} K^{2} L^{3} 
+ 608\,c_1^{2} c_2^{5} K^{2} L^{3} 
+ 1416\,c_1^{6} c_2 K^{2} L^{2} 
+ 5728\,c_1^{5} c_2^{2} K^{2} L^{2} 
- 27344\,c_1^{4} c_2^{3} K^{2} L^{2} 
+ 13408\,c_1^{3} c_2^{4} K^{2} L^{2} 
- 1400\,c_1^{2} c_2^{5} K^{2} L^{2} 
- 3744\,c_1^{6} c_2 K^{2} L 
- 81\,c_1^{6} L^{4} 
+ 128\,c_1^{5} c_2^{2} K^{2} L 
+ 342\,c_1^{5} c_2 L^{4} 
+ 53312\,c_1^{4} c_2^{3} K^{2} L 
- 559\,c_1^{4} c_2^{2} L^{4} 
- 18304\,c_1^{3} c_2^{4} K^{2} L 
+ 436\,c_1^{3} c_2^{3} L^{4} 
+ 1376\,c_1^{2} c_2^{5} K^{2} L 
- 159\,c_1^{2} c_2^{4} L^{4} 
+ 22\,c_1 c_2^{5} L^{4} 
- c_2^{6} L^{4} 
+ 2592\,c_1^{6} c_2 K^{2} 
+ 648\,c_1^{6} L^{3} 
- 5760\,c_1^{5} c_2^{2} K^{2} 
- 2736\,c_1^{5} c_2 L^{3} 
- 37696\,c_1^{4} c_2^{3} K^{2} 
+ 4472\,c_1^{4} c_2^{2} L^{3} 
+ 8576\,c_1^{3} c_2^{4} K^{2} 
- 3488\,c_1^{3} c_2^{3} L^{3} 
- 480\,c_1^{2} c_2^{5} K^{2} 
+ 1272\,c_1^{2} c_2^{4} L^{3} 
- 176\,c_1 c_2^{5} L^{3} 
+ 8\,c_2^{6} L^{3} 
- 1944\,c_1^{6} L^{2} 
+ 8208\,c_1^{5} c_2 L^{2} 
- 13416\,c_1^{4} c_2^{2} L^{2} 
+ 10464\,c_1^{3} c_2^{3} L^{2} 
- 3816\,c_1^{2} c_2^{4} L^{2} 
+ 528\,c_1 c_2^{5} L^{2} 
- 24\,c_2^{6} L^{2} 
+ 2592\,c_1^{6} L 
- 10944\,c_1^{5} c_2 L 
+ 17888\,c_1^{4} c_2^{2} L 
- 13952\,c_1^{3} c_2^{3} L 
+ 5088\,c_1^{2} c_2^{4} L 
- 704\,c_1 c_2^{5} L 
+ 32\,c_2^{6} L 
- 1296\,c_1^{6} 
+ 5472\,c_1^{5} c_2 
- 8944\,c_1^{4} c_2^{2} 
+ 6976\,c_1^{3} c_2^{3} 
- 2544\,c_1^{2} c_2^{4} 
+ 352\,c_1 c_2^{5} 
- 16\,c_2^{6}.
\end{autobreak}
\end{align*}

\end{theorem}

%ca*cb*k*l*(ca+cb)*(ca*k*l+cb*k*l-2*k-2*l)*(ca^2*cb*k*l+ca*cb^2*k*l-2*ca*cb*k-2*ca*cb*l+2*ca+2*cb)*(ca^2*cb*k*l+ca*cb^2*k*l-4*ca*cb*k-4*ca*cb*l+4*ca+4*cb)

%[[[1, 1/2, 1, 1], 1], [[1, 1/2, 3/2, 1], 1], [[1, 1/2, 5/2, 1], 1], [[1, 1/2, 11/4, 1], 1], [[1, 1/2, 7/2, 2], 1], [[1, 2, 1/2, 1], 1], [[1
%, 2, 3/4, 1], 1], [[1, 2, 5/4, 1], 1], [[1, 2, 11/8, 1], 1], [[1, 2, 7/4, 1], 1]], [[[1, 1/2, 1, 1], 1], [[1, 1/2, 3/2, 1], 1], [[1, 1/2, 5/2, 1], 1], [[1, 1/2, 11/4, 1], 1], [[1,
%1/2, 7/2, 2], 1], [[1, 2, 1/2, 1], 1], [[1, 2, 3/4, 1], 1], [[1, 2, 5/4, 1], 1], [[1, 2, 11/8, 1], 1], [[1, 2, 7/4, 1], 1]]

\begin{theorem}
For Model GA, if $C_1(q_1)=c_1q_1$ and $C_2(q_2)=c_2q_2$, there exists a unique equilibrium with $q_1,q_2>0$. Moreover, this equilibrium is locally stable if $R_{GA}^2>0$ and $R_{GA}^2<0$,
where
\begin{align*}
	R_{GA}^2 =&\, c_1^2KL + 2\,c_1c_2KL + c_2^2KL - 4\,c_1c_2K - 2\,c_1L - 2\,c_2L,\\
	R_{GA}^3
	 =& \,c_1^2KL + 2\,c_1c_2KL + c_2^2KL - 8\,c_1c_2K - 4\,c_1L - 4\,c_2L + 8\,c_1 + 8\,c_2.
\end{align*}
\end{theorem}

\begin{proposition}
	For Model GG, if $K_1=K_2$, the stable region for the linear costs $C_i(q_i)=c_iq_i$ is strictly contained in that for the quadratic costs $C_i(q_i)=c_iq_i^2$.
\end{proposition}

\begin{figure}[htbp]
  \centering
    \subfigure[$K=1$, $L=1/2$.]{\includegraphics[width=0.4\textwidth]{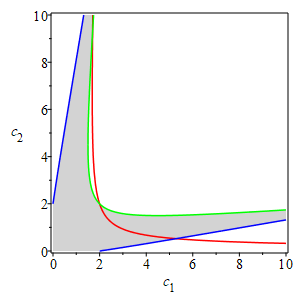}} 
    \subfigure[$K=2$, $L=1/2$.]{\includegraphics[width=0.4\textwidth]{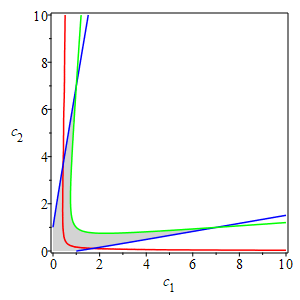}}\\
    \subfigure[$c_1=1$, $L=1/2$.]{\includegraphics[width=0.4\textwidth]{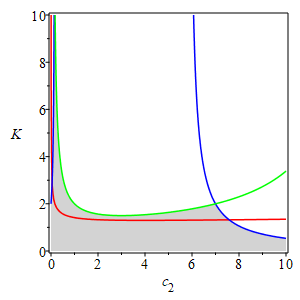}} 
    \subfigure[$c_2=1$, $K=1$.]{\includegraphics[width=0.4\textwidth]{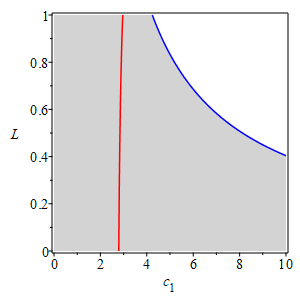}}\\
    \subfigure[$c_1=1$, $c_2=1$.]{\includegraphics[width=0.4\textwidth]{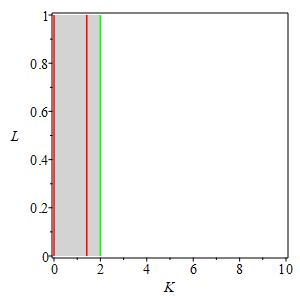}} 
    \subfigure[$c_1=1$, $c_2=5$.]{\includegraphics[width=0.4\textwidth]{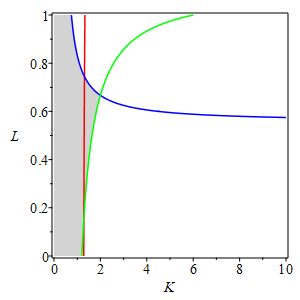}}\\ 
  \caption{The parameter space of Model GA. $R_{GA}^1=0$, $R_{GA}^2=0$ and $R_{GA}^3=0$ are marked in red, blue and green, respectively.}
    \label{fig:par-ga}
\end{figure}

\subsection{Model GG}

The Jacobian matrix is
\begin{equation}
J_{GG}=\left[
	\begin{matrix}
	1+K_1\cdot \partial G_1/\partial q_1 & K_1\cdot \partial G_1/\partial q_2\\
	K_2\cdot \partial G_2/\partial q_1 & 1+K_2\cdot \partial G_2/\partial q_2
	\end{matrix}
\right]
\end{equation}

Hence, the stable equilibria can be described by 
\begin{equation}\label{eq:semi-gg}
	\left\{\begin{split}
		&K_1\cdot G_1(q_1,q_2)=0,\\
		&K_2\cdot G_2(q_2,q_1)=0,\\
		&q_1>0,~q_2>0,\\
		&1+\Tr(J_{GG})+\Det(J_{GG})>0,\\
		&1-\Tr(J_{GG})+\Det(J_{GG})>0,\\
		&1-\Det(J_{GG})>0,\\
		&c_1>0,~c_2>0,~K_1>0,~K_2>0.
	\end{split}
	\right.
\end{equation}

% ca*cb*ka*kb*(ca-cb)*(ca+cb)*(-1024*ca^3*cb^3*ka^4*kb^4+384*ca^3*cb^2*ka^4*kb^2+384*ca^3*cb^2*ka^3*kb^3+384*ca^2*cb^3*ka^3*kb^3+384*ca^2*cb^3*ka^2*kb^4+ca^4*ka^4-18*ca^3
%*cb*ka^4-32*ca^3*cb*ka^3*kb-18*ca^3*cb*ka^2*kb^2+81*ca^2*cb^2*ka^4+288*ca^2*cb^2*ka^3*kb+420*ca^2*cb^2*ka^2*kb^2+288*ca^2*cb^2*ka*kb^3+81*ca^2*cb^2*kb^4-18*ca*cb^3*ka^2*kb^2-32*ca
%*cb^3*ka*kb^3-18*ca*cb^3*kb^4+cb^4*kb^4)*(-64*ca^4*cb^4*ka^4*kb^4+96*ca^4*cb^3*ka^4*kb^2-32*ca^4*cb^3*ka^3*kb^3-32*ca^3*cb^4*ka^3*kb^3+96*ca^3*cb^4*ka^2*kb^4+ca^5*cb*ka^4-18*ca^4*
%cb^2*ka^4+32*ca^4*cb^2*ka^3*kb-18*ca^4*cb^2*ka^2*kb^2+81*ca^3*cb^3*ka^4+96*ca^3*cb^3*ka^3*kb-92*ca^3*cb^3*ka^2*kb^2+96*ca^3*cb^3*ka*kb^3+81*ca^3*cb^3*kb^4-18*ca^2*cb^4*ka^2*kb^2+
%32*ca^2*cb^4*ka*kb^3-18*ca^2*cb^4*kb^4+ca*cb^5*kb^4+8*ca^4*cb*ka^2-8*ca^4*cb*ka*kb-16*ca^3*cb^2*ka^2-120*ca^3*cb^2*ka*kb-120*ca^3*cb^2*kb^2-120*ca^2*cb^3*ka^2-120*ca^2*cb^3*ka*kb-\
%16*ca^2*cb^3*kb^2-8*ca*cb^4*ka*kb+8*ca*cb^4*kb^2-4*ca^4+16*ca^3*cb-24*ca^2*cb^2+16*ca*cb^3-4*cb^4)*(-1024*ca^4*cb^4*ka^4*kb^4+384*ca^4*cb^3*ka^4*kb^2-640*ca^4*cb^3*ka^3*kb^3-640*
%ca^3*cb^4*ka^3*kb^3+384*ca^3*cb^4*ka^2*kb^4+ca^5*cb*ka^4-18*ca^4*cb^2*ka^4+96*ca^4*cb^2*ka^3*kb-210*ca^4*cb^2*ka^2*kb^2+81*ca^3*cb^3*ka^4-96*ca^3*cb^3*ka^3*kb-732*ca^3*cb^3*ka^2*
%kb^2-96*ca^3*cb^3*ka*kb^3+81*ca^3*cb^3*kb^4-210*ca^2*cb^4*ka^2*kb^2+96*ca^2*cb^4*ka*kb^3-18*ca^2*cb^4*kb^4+ca*cb^5*kb^4+8*ca^4*cb*ka^2-40*ca^4*cb*ka*kb-16*ca^3*cb^2*ka^2-88*ca^3*
%cb^2*ka*kb-120*ca^3*cb^2*kb^2-120*ca^2*cb^3*ka^2-88*ca^2*cb^3*ka*kb-16*ca^2*cb^3*kb^2-40*ca*cb^4*ka*kb+8*ca*cb^4*kb^2-4*ca^4+16*ca^3*cb-24*ca^2*cb^2+16*ca*cb^3-4*cb^4)

\begin{theorem}
For Model GG, there exists a unique equilibrium with $q_1,q_2>0$. Moreover, this equilibrium is locally stable if $R_{GG}^1>0$ and $R_{GG}^2<0$,
where
\begin{align*} 
R_{GG}^1 =\,& -1024\, c_1^3 c_2^3 K_1^4 K_2^4+384\, c_1^3 c_2^2 K_1^4 K_2^2+384\, c_1^3 c_2^2 K_1^3 K_2^3+384\, c_1^2 c_2^3 K_1^3 K_2^3+384\, c_1^2 c_2^3 K_1^2 K_2^4+c_1^4 K_1^4\\
&-18\, c_1^3 c_2 K_1^4-32 c_1^3 c_2 K_1^3 K_2-18\, c_1^3 c_2 K_1^2 K_2^2+81\, c_1^2 c_2^2 K_1^4+288\, c_1^2 c_2^2 K_1^3 K_2+420\, c_1^2 c_2^2 K_1^2 K_2^2\\
&+288\, c_1^2 c_2^2 K_1 K_2^3+81\, c_1^2 c_2^2 K_2^4-18\, c_1 c_2^3 K_1^2 K_2^2-32 c_1 c_2^3 K_1 K_2^3-18\, c_1 c_2^3 K_2^4+c_2^4 K_2^4,
\end{align*}
and
\begin{align*}
R_{GG}^2 =\,& -64\, c_1^4 c_2^4 K_1^4 K_2^4+96\, c_1^4 c_2^3 K_1^4 K_2^2-32\, c_1^4 c_2^3 K_1^3 K_2^3-32\, c_1^3 c_2^4 K_1^3 K_2^3+96\, c_1^3 c_2^4 K_1^2 K_2^4+c_1^5 c_2 K_1^4\\
&-18\, c_1^4 c_2^2 K_1^4+32\, c_1^4 c_2^2 K_1^3 K_2-18\, c_1^4 c_2^2 K_1^2 K_2^2+81\, c_1^3 c_2^3 K_1^4+96\, c_1^3 c_2^3 K_1^3 K_2-92\, c_1^3 c_2^3 K_1^2 K_2^2\\
&+96\, c_1^3 c_2^3 K_1 K_2^3+81\, c_1^3 c_2^3 K_2^4-18\, c_1^2 c_2^4 K_1^2 K_2^2+32\, c_1^2 c_2^4 K_1 K_2^3-18\, c_1^2 c_2^4 K_2^4+c_1 c_2^5 K_2^4+8\, c_1^4 c_2 K_1^2\\
&-8\, c_1^4 c_2 K_1 K_2-16 c_1^3 c_2^2 K_1^2-120\, c_1^3 c_2^2 K_1 K_2 -120\, c_1^3 c_2^2 K_2^2-120\, c_1^2 c_2^3 K_1^2-120\, c_1^2 c_2^3 K_1 K_2\\
&-16 \,c_1^2 c_2^3 K_2^2-8\, c_1 c_2^4 K_1 K_2+8\, c_1 c_2^4 K_2^2-4 \,c_1^4+16\,c_1^3c_2-24\,c_1^2c_2^2+16\,c_1c_2^3-4\,c_2^4.\\
\end{align*}
%\begin{align*}
%R_{GG}^1 =& -1024*ca^3*cb^3*k^4*l^4+384*ca^3*cb^2*k^4*l^2+384*ca^3*cb^2*k^3*l^3+384*ca^2*cb^3*k^3*l^3+384*ca^2*cb^3*k^2*l^4+ca^4*k^4-18*ca^3*cb*k^4-32*ca^3*cb*k^3*l-18*ca^3*cb*k^2*l^2+81*ca^2*cb^2*k^4+288*ca^2*cb^2*k^3*l+420*ca^2*cb^2*k^2*l^2+288*ca^2*cb^2*k*l^3+81*ca^2*cb^2*l^4-18*ca*cb^3*k^2*l^2-32*ca*cb^3*k*l^3-18*ca*cb^3*l^4+cb^4*l^4,\\
%R_{GG}^2 =& -64*ca^4*cb^4*k^4*l^4+96*ca^4*cb^3*k^4*l^2-32*ca^4*cb^3*k^3*l^3-32*ca^3*cb^4*k^3*l^3+96*ca^3*cb^4*k^2*l^4+ca^5*cb*k^4-18*ca^4*cb^2*k^4+32*ca^4*cb^2*k^3*l-18*ca^4*cb^2*k^2*l^2+81*ca^3*cb^3*k^4+96*ca^3*cb^3*k^3*l-92*ca^3*cb^3*k^2*l^2+96*ca^3*cb^3*k*l^3+81*ca^3*cb^3*l^4-18*ca^2*cb^4*k^2*l^2+32*ca^2*cb^4*k*l^3-18*ca^2*cb^4*l^4+ca*cb^5*l^4+8*ca^4*cb*k^2-8*ca^4*cb*k*l-16*ca^3*cb^2*k^2-120*ca^3*cb^2*k*l-120*ca^3*cb^2*l^2-120*ca^2*cb^3*k^2-120*ca^2*cb^3*k*l-16*ca^2*cb^3*l^2-8*ca*cb^4*k*l+8*ca*cb^4*l^2-4*ca^4+16*ca^3*cb-24*ca^2*cb^2+16*ca*cb^3-4*cb^4.\\
%\end{align*}

\end{theorem}

%ca*cb*k*l*(ca+cb)*(ca*k*l+cb*k*l-2*k-2*l)*(ca^2*cb*k*l+ca*cb^2*k*l-2*ca*cb*k-2*ca*cb*l+2*ca+2*cb)*(ca^2*cb*k*l+ca*cb^2*k*l-4*ca*cb*k-4*ca*cb*l+4*ca+4*cb)

%[[[1, 1/2, 1, 1], 1], [[1, 1/2, 3/2, 1], 1], [[1, 1/2, 5/2, 1], 1], [[1, 1/2, 11/4, 1], 1], [[1, 1/2, 7/2, 2], 1], [[1, 2, 1/2, 1], 1], [[1
%, 2, 3/4, 1], 1], [[1, 2, 5/4, 1], 1], [[1, 2, 11/8, 1], 1], [[1, 2, 7/4, 1], 1]], [[[1, 1/2, 1, 1], 1], [[1, 1/2, 3/2, 1], 1], [[1, 1/2, 5/2, 1], 1], [[1, 1/2, 11/4, 1], 1], [[1,
%1/2, 7/2, 2], 1], [[1, 2, 1/2, 1], 1], [[1, 2, 3/4, 1], 1], [[1, 2, 5/4, 1], 1], [[1, 2, 11/8, 1], 1], [[1, 2, 7/4, 1], 1]]

\begin{theorem}
For Model GG, if $C_1(q_1)=c_1q_1$ and $C_2(q_2)=c_2q_2$, there exists a unique equilibrium with $q_1,q_2>0$. Moreover, this equilibrium is locally stable if $R_{GG}^3>0$ and $R_{GG}^4<0$,
where
\begin{align*}
	R_{GG}^3 =&\, c_1K_1K_2+c_2K_1K_2-2\,K_1-2\,K_2,\\
	R_{GG}^4 =&\, c_1^2c_2K_1K_2+c_1c_2^2K_1K_2-4\,c_1c_2K_1-4\,c_1c_2K_2+4\,c_1+4\,c_2.
\end{align*}
\end{theorem}

\begin{proposition}
	For Model GG, if $K_1=K_2$, the stable region for the linear costs $C_i(q_i)=c_iq_i$ is strictly contained in that for the quadratic costs $C_i(q_i)=c_iq_i^2$.
\end{proposition}

\begin{proposition}
	For Model GG, even if $c_1>10^{100}$ and $c_2>10^{100}$, the stable region for the linear costs $C_i(q_i)=c_iq_i$ is not strictly contained in that for the quadratic costs $C_i(q_i)=c_iq_i^2$.
\end{proposition}

\begin{figure}[htbp]
  \centering
     \subfigure[$K_1=K_2$.]{\includegraphics[width=0.4\textwidth]{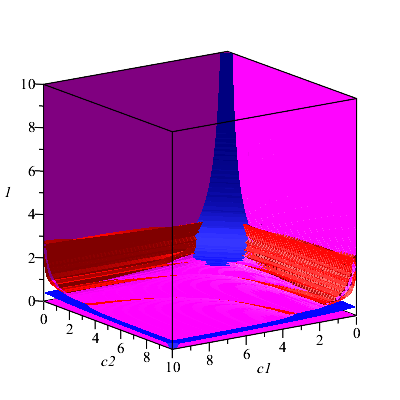}} 
    \subfigure[$K_1=K_2$.]{\includegraphics[width=0.4\textwidth]{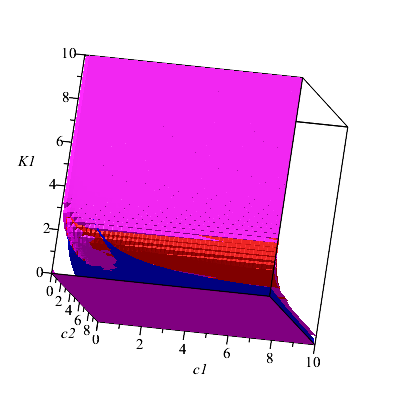}}\\
    \subfigure[$K_1=K_2$, $c_1=c_2$.]{\includegraphics[width=0.4\textwidth]{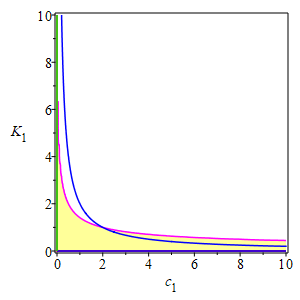}} 
    \subfigure[$K_1=K_2=1$.]{\includegraphics[width=0.4\textwidth]{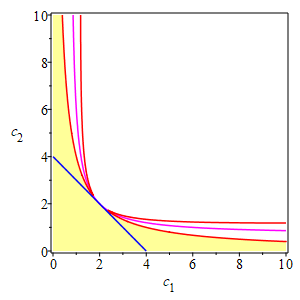}}\\
    \subfigure[$c_1=1$, $c_2=1$.]{\includegraphics[width=0.4\textwidth]{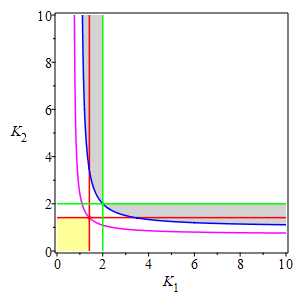}} 
    \subfigure[$c_1=1$, $c_2=1/2$.]{\includegraphics[width=0.4\textwidth]{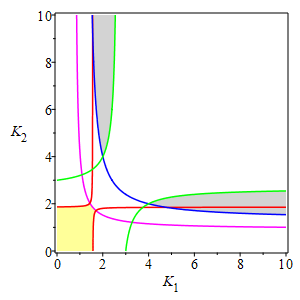}}\\
  \caption{The parameter space of Model GG. The surfaces $R_{GG}^1=0$, $R_{GG}^2=0$, $R_{GG}^3=0$ and $R_{GG}^4=0$ are marked in magenta, red, blue and green, respectively.}
    \label{fig:par-gg}
    
\end{figure}

%\section{Bifurcation Analysis}
%
%\section{Concluding Remarks}
%
%\section*{Acknowledgements}
%
%The authors are grateful to the anonymous referees for their helpful comments. This work has been supported by Philosophy and Social Science Foundation of Guangdong under Grant GD21CLJ01, Major Research and Cultivation Project of Dongguan City College under Grant 2021YZDYB04Z and Social Development Science and Technology Project of Dongguan under Grant 20211800900692. 
%

\bibliographystyle{abbrv}
\bibliography{duopoly.bib}

%\section*{Appendix}
%\tiny
%\begin{align*}
%
%\end{autobreak}\\
%
%\begin{autobreak}

%\end{align*}

%
%\newpage
%
%\section*{Highlights}
%
%\begin{enumerate}
%    \item The complete condition of the local stability for the modified version of Puu's monopoly model is obtained for the first time.
%    \item We investigate periodic solutions as well as their stability rigorously through symbolic computations rather than intuitively through numerical simulations.
%    \item The existence of chaos in the sense of Li-Yorke is proved by finding snapback repellers and 3-cycle orbits respectively for the two models.
%\end{enumerate}
%
%
%
%\section*{Cover Letter}
%\begin{lstlisting}[breaklines=true, columns=flexible]
%
%Dear editor,
%
%I would like to submit the enclosed manuscript entitled "***" by *** for possible publication in Communications in ***. 
%
%[abstract]
%
%Highlights of this work include the following.
%
%[highlights]
%
%The contact information of the corresponding author is as follows.
%
%   Name: ***;
%   Address: ***; 
%   E-mail: ***;
%   Mobile: ***;
%   Fax: N/A.
%
%Thank you very much for consideration!
%
%Sincerely yours,
%***
%
%\end{lstlisting}

\end{CJK}
\end{document}

%% file: macro.tex
\usepackage{amssymb}
\usepackage{amsmath}
\usepackage{amsthm}
\usepackage{amsfonts}
\usepackage{mathtools}
\usepackage{mathrsfs}
\usepackage{bm} % 处理数学公式中的黑斜体的宏包

\usepackage{graphicx}
\usepackage{dcolumn}
\usepackage[ruled,lined,algonl]{algorithm2e}
\usepackage{color}
\usepackage{autobreak}
\allowdisplaybreaks

\theoremstyle{plain}
\newtheorem{theorem}{Theorem}

\newtheorem{proposition}{Proposition}

\theoremstyle{definition}

\newtheorem{model}{Model}

\theoremstyle{remark}

\DeclareMathOperator{\Det}{Det}
\DeclareMathOperator{\Tr}{Tr}